%% file: paper.tex
\documentclass[twoside,leqno,twocolumn]{article}

\input{macros}

\usepackage{ltexpprt}

\usepackage{cite}

\usepackage{enumitem}

\usepackage{stfloats}
\usepackage{graphicx}
\usepackage{cite}
\usepackage{times}
\usepackage{latexsym}
\usepackage{graphicx}
\usepackage{amsmath}
\usepackage{amssymb}
\usepackage{amsfonts}
\usepackage{algorithm}
\usepackage{algorithmic}
\usepackage{url}
\usepackage{bm}
\usepackage{mdwlist,lipsum}
\usepackage{float}
\usepackage{color}
\usepackage{subfigure}
\usepackage{multirow}

\usepackage{lipsum}
\makeatletter
\newcommand\resetsubfigs{\setcounter{sub\@captype}{0}}
\makeatother

\newfloat{algorithm}{htb}{loa}
\floatname{algorithm}{Algorithm}

\newtheorem{Theorem}{Theorem}

\newcommand{\coord}{\bm{\alpha}}
\newcommand{\tcoord}{\bm{\alpha}}

\newcommand{\ComAdjVec}[1]{\bm{x}_{C_{#1}}}

\newcommand{\pertMatrix}{\widetilde{A}}

\newcommand{\pertAdjVec}{\tilde{\bm{x}}}

\newcommand{\kmeans}{$k$-means algorithm }

\usepackage{color}

\begin{document}

\title{On Spectral Analysis of Directed Signed Graphs}

\author{%
{Yuemeng Li{\small $~^{1}$}, Xintao Wu{\small $~^{2}$}, Aidong Lu{\small $~^{1}$}}%
\\
$^{1}$\, University of North Carolina at Charlotte, USA, \{yli60,alu1\} @uncc.edu  \\
$^{2}$\, University of Arkansas, USA, xintaowu@uark.edu
}

\date{}

\maketitle

\begin{abstract} \small\baselineskip=9pt
It has been shown that the adjacency eigenspace of a network contains key information of its underlying structure. However, there has been no study on spectral analysis of the adjacency matrices of directed signed graphs. In this paper, we  derive theoretical approximations of spectral projections from such directed signed networks using matrix perturbation theory.  We use the derived theoretical results to study the influences of negative intra cluster and inter cluster  directed edges on node spectral projections. We then develop a spectral clustering based graph partition algorithm, SC-DSG, and conduct evaluations on both synthetic and real datasets. Both theoretical analysis and empirical evaluation demonstrate the effectiveness of the proposed algorithm.
\end{abstract}

%\keywords{Directed graphs; Signed graphs; Asymmetric adjacency matrices; Matrix perturbation; Spectral projection; Graph partition; Handle complex eigenpairs}

\section{Introduction}

In social networks, relationships between two individuals are often directed, such as Twitter following, phone calls, and voting. Directed graphs are used to capture asymmetric relationships between individuals.
Spectral properties for directed graphs have  been studied abundantly in the past (refer to the surveys \cite{Seary2003,aouchiche2010survey,Malliaros2013}).

Relationships in social networks could  have more than  two status like presence or absence of a trust/friendship between two individuals. They could also be negative to express distrust or dislike. Signed networks are used for this purpose. Spectral analysis of signed graphs have also been studied \cite{Zaslavsky2008}. For example, the work \cite{hou2005bounds} studied the bounds for the least Laplacian eigenvalue of a signed graph. A later work \cite{kunegis2009learning} used spectral properties of the signed graph for link prediction. The work \cite{Traag2009} extended the modularity metric for unsigned graphs to the signed modularity for signed graphs. The authors in \cite{kunegis2010spectral} studied the spectral properties of signed normalized Laplacian transformation from the original signed adjacency matrix and developed methods for spectral clustering, link prediction and graph structure visualization.

However, spectral analysis of directed signed graphs (DSGs) has not been thoroughly studied, since previous works focused on either unsigned graphs or undirected signed graphs.
In the ideal case of DSGs, all the intra community edges are positive and all the inter community edges are negative since the members within one community tend to hold the same opinion towards each other while members from different communities tend to dispute.
However, in real world datasets such as Epinion, negative links are also present within communities and some positive links are present between communities.

It was shown in \cite{Wu2011} and \cite{li2015analysis} that matrix perturbation theories can be used as a powerful tool for explaining the effects of inter community edges on the spectral projection behaviours of the given adjacency matrix directly. The former work provided theoretical results for undirected graphs, while the later work conducted theoretical analysis for directed graphs. In \cite{Wu2011b}, the authors analyzed the $K$-balanced undirected signed graphs by using matrix perturbation approach. In a recent work \cite{wu2014spectral}, the effects of negative edges on the spectral properties of signed and dispute networks were studied. However, the influences of negative edges to the spectral properties of DSGs remain unclear, so many problems in this domain are still open.

The core idea of applying the matrix perturbation theories on spectral graph analysis is to model the observed graph (with $K$ communities) as the perturbation of intra-community edges on a $K$-block graph (with $K$ disconnected communities) and study how the spectral space formed by leading eigenvectors as well as node projections in the space are changed before and after perturbation. However, when applying the matrix perturbation theories on DSGs, one main difficulty is to deal with the complex eigenpairs associated with the asymmetric adjacency matrix. In \cite{li2015analysis}, the authors utilized the strong-connectedness property of the communities  and the real Perron-Frobenius eigenvalue and eigenvector of each community, thus eliminating the need for dealing with complex eigenpairs. However, when the graphs have negative intra cluster edges, the Perron-Frobenuius eigenpairs may not be real any more. In this paper, we propose to handle the inter cluster and intra cluster negative entries of DSGs separately.

We apply matrix perturbation theories to derive several key theoretical results for analyzing negative inter cluster edges. Our key results can answer the following important questions: How will the negative intra cluster edges affect the spectral projections of each node? Will negatively linked nodes be pushed away from each other, while positively linked nodes be pulled towards each other like those in undirected signed graphs? What is the role of the directionality of an edge on node spectral coordinates? For negative intra cluster perturbation, we study how to deal with complex eigenpairs for DSGs.  We analyze  how intra cluster negative edges affect the Perron Frobenius eigenpair of the corresponding clusters. In particular, we explain why negative edges change real eigenpairs to complex eigenpairs based on the linear algebra results on the nonnegative irreducible and eventually non-negative irreducible states.  These questions are crucial in identifying the spectral properties of cluster relationships and developing spectral clustering algorithm for DSGs.
We conduct evaluations on several synthetic datasets and real networks and compare the accuracy results with several state-of-the-art spectral clustering methods. Results demonstrate the effectiveness of the proposed method.

To summarize, our paper has the following contributions. We explore the spectral space of the adjacency matrix for the DSG and derive theoretical results about the change of the spectral coordinate due to edges with different signs. To our best knowledge, this is the first work on exploring the spectral properties of the DSG. We propose a method to handle complex eigenpairs when spectral radii do not exist as eigenvalues. The proposed algorithm  can detect clusters with opposing relationships as well as clusters that are structurally separated from each other but with neutral relationships.

\section{Preliminaries}\label{sect:BG}

A directed signed graph with $n$ nodes can be represented as its adjacency matrix $A_{n\times n}$ with $A_{ij} = 1$ (-1)if there exists a positive (negative) edge pointing from node $v_i$ to node $v_j$ and $A_{ij} = 0$ otherwise. Since  $A_{ij}$ and $A_{ji}$ may not have the same value, $A$ is asymmetric.

\subsection{Spectral Projections}

The spectral decomposition of $A$ takes the form $A=\sum_{i}\lambda_i\bm{x}_i\bm{x}_i^T$.
When the graph is undirected, all the eigenvalues $(\lambda_1,\cdots, \lambda_n)$ are real and are assumed to be in descending order. The eigenvectors are sorted accordingly.
 \begin{equation}\label{illustration}
\begin{array}{@{}r@{}c@{}c@{}}
&\phantom{w}\begin{array}{@{}ccc}\bm{x}_1\phantom{iii}\phantom{\cdots}&\bm{x}_{i}&\phantom{\cdots}\phantom{iii}\bm{x}_{K}\end{array}& \begin{array}{cc}&\phantom{b}\bm{x}_n\end{array}\\
&\phantom{b}\downarrow &\\
\bm{\alpha}_{u}\rightarrow & \left( \begin{array}{l|c|r} x_{11}\ \cdots & x_{i1} & \cdots\ x_{K1}\\
\phantom{ij}\vdots        & \vdots &        \vdots\phantom{ij}\\
\hline x_{1u}\ \cdots & x_{iu} & \cdots\ x_{Ku}\\ \hline
\phantom{ij}\vdots        & \vdots &        \vdots\phantom{ij}\\
x_{1n}\ \cdots & x_{in} & \cdots\ x_{Kn}\\
\end{array}\right. &
\left.\begin{array}{@{}cc@{}}
\cdots & x_{n1}\\
       & \vdots\\
\cdots & x_{nu}\\
       & \vdots\\
\cdots & x_{nn}\\
\end{array}\right)\\
\end{array}
\end{equation}
The basis of the spectral space are formed by eigenvectors of the given adjacency matrix. The spectral space is of full rank $n$, when all the eigenvectors are linearly independent. If each row is treated as a coordinate in the $n$ dimensional space, then all the nodes can be projected into such a spectral space as shown in equation \ref{illustration}.  In most application, only the first $K$ eigenpairs contain major topological information. The row vector $\coord_u =(x_{1u}, x_{2u},\cdots,x_{Ku})$ are the coordinates used for clustering in this spectral subspace.

Unlike undirected graphs, the eigenvectors of DSGs do not form orthonormal basis directly.  The work \cite{li2015analysis} suggested that the $K$ eigenpairs form the perturbed Perron Frobenius simple invariant subspace for directed graphs and each eigenpair corresponds to one strongly connected community which is irreducible and non-negative.

\begin{Definition}\label{def:Char_Ploy}
The characteristic polynomial of a $n$ by $n$ matrix $A$ takes the general form:
\begin{equation}\label{eq:Char_Poly}
F(\lambda)=a_{1}*\lambda^{n}+a_{2}*\lambda^{n-1}+\cdots+a_{n}*\lambda+a_{n+1},
\end{equation}
where the roots for $F(\lambda)=0$ will be the eigenvalues of $A$.
\end{Definition}

\begin{Definition}\label{def:Spectral_Radius}
Let $\Lambda=(\lambda_{1}, \cdots, \lambda{n})$ be the eigenvalues of matrix $A$, then $\rho(A)=max(|\Lambda|)$ is called the spectral radius of $A$. In the case where complex valued eigenvalues exist, the absolute values become the moduli.
\end{Definition}

However, when negative intra community edges are included in DSGs, the eigenpair corresponding to the spectral radius could be complex. In order to avoid confusion, in this paper we assume that the top eigenvectors are sorted based on the moduli so that the first $K$ eigenvectors still correspond to the ones that form the perturbed Perron Frobenius simple invariant subspace.

\subsection{$K$-block Graph}

In matrix perturbation based spectral graph analysis, an observed graph $\tilde{A}$ with $K$ clusters is modeled as a $K$-block matrix $A$ (after permutation) perturbed by inter cluster edges $E$.
\begin{equation}\label{kcomm}
\widetilde{A} = A+E = \left(
                         \begin{array}{ccc}
                           A_1 &  & 0 \\
                            & \ddots &  \\
                           \mathbf{0} &  & A_K \\
                         \end{array}
                       \right) +E,
\end{equation}

According to the work \cite{li2015analysis}, for directed unsigned graphs, each cluster $C_i$ is assumed to be strongly connected. The dominant eigenvalue $\lambda_1$ of each component $C_i$ is positive, simple and the corresponding eigenvector $\mathbf{x_1}$ is positive.  If we choose $\ComAdjVec{i}$ to be the Perron-Frobenius eigenvectors of corresponding communities, then the eigenvectors $\mathbf{x}=(\bm{x}_1, \cdots, \bm{x}_K)$ of $A$ corresponding to $\lambda_{Ci}$s are the only eigenvectors whose non-zero components are all positive, all the entries of $\mathbf{x}$ are real valued and have the following form:
\begin{equation}\label{kcommCoord}
 (\bm{x}_1,\bm{x}_2,\cdots, \bm{x}_K) =
 \begin{pmatrix}
  \ComAdjVec{1} & \bm{0}   & \cdots & \bm{0} \\
  \bm{0}   & \ComAdjVec{2} & \cdots & \bm{0} \\
  \vdots   & \vdots   & \ddots & \vdots  \\
  \bm{0}   & \bm{0}   & \cdots & \ComAdjVec{K}
 \end{pmatrix}
\end{equation}

There is only one location of the row vector $\bm{\alpha_u}$ that has a non-zero value with the form:
\begin{equation}\label{eq:spctr-coord-0}
\bm{\alpha_u}=(0,\cdots,0,x_{iu},0,\cdots,0).
\end{equation}
The location of $x_{iu}$ indicates the $i-$th community which node $u$ belongs to and the value of $x_{iu}$ denotes the influence of node $u$ to that community.

\section{Spectral Analysis of DSGs}

For DSGs,  both community $C_i$'s and $E$ can contain negative edges.
We treat both positive and negative inter cluster edges as inter cluster perturbation and treat the negative edges within each cluster as intra cluster perturbation.
Formally, we have
\begin{equation}\label{kcomm}
\widetilde{A} = A+E_I +E_O
 \end{equation}
where $A$ is a $K$-block matrix as the same defined in Equation \ref{kcomm} and each diagonal component $A_i$ is still assumed strongly connected and nonnegative, $E_I$ is a $K$-block matrix corresponding to intra cluster perturbation and each diagonal component $E_i$ contains negative intra cluster edges, and $E_O$ contains both positive and negative inter cluster edges.
 We discuss these two situations and  derive several theoretical results to explain the associated spectral projection behaviour in Sections \ref{sect:Inter} and \ref{sect:Intra} respectively .

\subsection{Spectral Analysis of Inter Cluster Perturbation}\label{sect:Inter}

In this case,  our model is simplified as $\widetilde{A} = A+E_O$.
In \cite{li2015analysis}, the authors  studied the spectral properties of  directed unsigned graphs based on the matrix perturbation theories \cite{Stewart1990} and works \cite{Stewart1971,Stewart1993}.
Because the eigenvectors of asymmetric matrices do not form an orthonormal basis naturally, they developed a method of constructing orthonormal basis and derived the approximations of the eigenvectors when treating the graph as a perturbation from a block matrix.  The derived theories in \cite{li2015analysis} can be generalized to our DSG setting although although $E_O$ contains both positive and negative inter cluster edges. This is because both models assume each $A_i$ strongly connected and nonnegative, thus having the Perron-Frobenius simple invariant subspace. We refer their results below and then focus on how positive and negative edges in $E_O$ affect the spectral coordinates. To be consistent, in the remaining part of Section \ref{sect:Inter}, $E$ denotes $E_O$.

\begin{Theorem}\label{thm:perturbed-subspace-approx}
Let the observed graph be $\pertMatrix=A+E$ with $K$ communities and the perturbation $E$ denotes the edges connecting communities $C_1, \cdots, C_K$.
Let $(\bm{x}_1,\cdots,\bm{x}_K)$ be the relabeled Perron-Frobenius eigenvectors of $A$ for all communities, and $Q$ be the rest of the orthronormal basis constructed using Gram-Schmidt process. Then $(\bm{x}_1,\cdots,\bm{x}_K)$ is a simple invariant subspace of $A$, and the perturbed Perron-Frobenius spectral space for $\widetilde{A}$  can be approximated as:
\begin{equation}\label{eq:perturbed-subspace-approx}
(\pertAdjVec_1,\cdots, \pertAdjVec_K)  \approx (\bm{x}_1,\cdots,\bm{x}_K) + \nabla E(\frac{\bm{x}_1}{\lambda_1}, \cdots, \frac{\bm{x}_K}{\lambda_K}).
\end{equation}
where $\nabla=Q(I-\frac{L_2}{\lambda_i})^{-1}Q^H$.
\end{Theorem}

The theorem could be used to derive the approximation of spectral coordinate of $\coord_u$ using the following simplified result that only takes into account of the influences of neighboring nodes from other communities. Since the edge direction indicates the flow of information, we define the outer community neighbours of a node $u \in C_i$ to be any node $v \notin C_i$ that has an edge pointing to $u$.

\begin{Theorem}\label{thm:adj-coord-approx}
For node $u\in C_i$, let $\Gamma^j_u$ denote its set of neighbors in $C_j$ for $j \in (1, \cdots, K)$. The simplified spectral coordinates $\tcoord_u$ can be approximated as:
\begin{equation}\label{eq:spctr-coord-approx}
\resizebox{.85\hsize}{!}{$ \tcoord_u\approx x_{iu}I_{i} + \left(\sum^{n}_{j=1}\nabla_{uj}\sum_{v\in\Gamma^1_u}\frac{e_{jv}x_{1v}}{\lambda_1},\cdots,
  \sum^{n}_{j=1}\nabla_{uj}\sum_{v\in\Gamma^K_u}\frac{e_{jv}x_{Kv}}{\lambda_K}\right)$},
\end{equation}
where $I_i$ is the $i$-th row of a $K$-by-$K$ identity matrix, $e_{jv}$ is the $(j,v)$ entry of E and $\nabla$ is defined in Theorem \ref{thm:perturbed-subspace-approx}.
\end{Theorem}

The work in \cite{li2015analysis} only gave the above approximation formula and did not examine how the node spectral coordinates change under perturbation of inter cluster edges.  One reason is that the entry $\sum^{n}_{j=1}\nabla_{uj}\sum_{v\in\Gamma^i_u}\frac{e_{jv}x_{iv}}{\lambda_i}$ in the $i$-th column position of the spectral coordinate in Equation \eqref{eq:spctr-coord-approx} is very complicate compared with that of undirected graphs and hence it is difficult to determine the influence of the perturbation. In this work, we decompose the perturbation into each edge and explicitly quantify how one single inter cluster edge $u \rightarrow v$ changes the spectral coordinates of $u$ and $v$.

Without loss of generality, suppose nodes $u$ and $v$ are from community $C_1$ and $C_2$ respectively, there is a directed edge from $u$ to $v$, $u \rightarrow v$, which could be positive or negative.
Before the edge added, the spectral coordinates for nodes $u$ and $v$ in the two dimensional space are
  $     \left(
              \begin{array}{cc}
              x_{1u} & 0 \\
               \mathbf{0} & x_{2v} \\
                \end{array}
       \right).$
After the edge added, from Theorem \ref{thm:adj-coord-approx}, the spectral coordinates are
$       \left(
              \begin{array}{cc}
              x_{1u} & \nabla_{u1}\frac{e_{uv}}{\lambda_2}x_{2v} \\
               \mathbf{0} & x_{2v}+\nabla_{v1}\frac{e_{uv}}{\lambda_2}x_{2v} \\
                \end{array}
       \right).$

Our next theorem shows that the change of spectral coordinates depends on both the Perron-Frobenius eigenvalue of the node's community and the edge directionality.

\begin{Theorem}\label{thm:Generalized_Rotation}
Denote $(\lambda_1, \bm{x}_1)$ and $(\lambda_2, \bm{x}_2)$ as the Perron-Frobenius eigenpair of $C_1$ and $C_2$ respectively.  Nodes $u$ and $v$ are from community $C_1$ and $C_2$ respectively.
\begin{enumerate}
\item When  $u \rightarrow v$ is positive,
\begin{enumerate}[label=(\alph*)]
\item If $\lambda_1>\lambda_2$, node $u$ has a clockwise rotation while node $v$ stays on its original axis.
\item If $\lambda_1<\lambda_2$, node $u$ stays on its original axis while node $v$ has a clockwise rotation.
\end{enumerate}

\item When  $u \rightarrow v$ is negative,
\begin{enumerate}[label=(\alph*)]
\item If $\lambda_1>\lambda_2$, node $u$ has an anti-clockwise rotation while node $v$ stays on its original axis.
\item If $\lambda_1<\lambda_2$, node $u$ stays on its original axis while node $v$ has an anti-clockwise rotation.
\end{enumerate}
\end{enumerate}
\end{Theorem}
\begin{proof}
For 1(a), node $v$ has spectral coordinate $(0, {x}_{2v}+\nabla_{v1}\frac{e_{uv}}{\lambda_2}x_{2v})$. Therefore, node $v$ will stay on its original axis. On the other hand, node $u$ has spectral coordinate $({x}_{1u}, \nabla_{u1}\frac{e_{uv}}{\lambda_2}x_{2v})$. The angle $\beta$ of the spectral coordinate vector of node $u$ with the ${x}_1$ axis will be $\arctan(\frac{\nabla_{u1}\frac{e_{uv}}{\lambda_2}x_{2v}}{{x}_{1u}})$. The top part $\nabla_{u1}*\frac{1}{\lambda_2}$ takes the full form as $(Y(\lambda_2 I-L_2)^{-1}Y^H)_{u1}$ as in Theorem \ref{thm:perturbed-subspace-approx}.
The diagonal of $L_2$ are the other eigenvectors of $A$ by the construction process. Furthermore, $L_2$ itself is upper triangular by Schur's Theorem. Then the diagonal entries of $(\lambda_2 I-L_2)^{-1}$ becomes $(\lambda_2-\lambda_1,\lambda_2-\lambda_3,\cdots, \lambda_2-\lambda_n)^{-1}$. If we divide $\nabla$ by the term $(\lambda_2-\lambda_1)^{-1}$ and relabel it as $\nabla^*$, then the spectral coordinates of $u$ becomes $({x}_{iu},(\lambda_2-\lambda_1)^{-1}\nabla^{*}_{u1}e_{uv}x_{2v})$. Then the angle $\beta$ becomes $\arctan(\frac{(\lambda_2-\lambda_1)^{-1}\nabla^{*}_{u1}e_{uv}x_{2v}}{{x}_{1u}})$. Since $\lambda_2-\lambda_1<0$, $\beta$ will be a negative angle, which indicates that node $u$ will rotate clockwise to the fourth quadrant.

For 1(b), if $\lambda_1<\lambda_2$, by relabeling $\nabla^{*}$ if necessary, the angle $\beta$ takes the same equation as $\arctan(\frac{(\lambda_2-\lambda_1)^{-1}\nabla^{*}_{u1}e_{uv}x_{2v}}{{x}_{1u}})$. Since $\lambda_1<\lambda_2$, $\beta$ will be a positive angle, which indicates that node $u$ will rotate counter-clockwise and the vector will remain in the first quadrant.

For 2(a), $e_{uv}<0$ and $\lambda_1>\lambda_2$.  By relabeling $\nabla^{*}$ if necessary, the angle $\beta$ takes the same equation as $\arctan(\frac{(\lambda_2-\lambda_1)^{-1}\nabla^{*}_{u1}e_{uv}x_{2v}}{{x}_{1u}})$. Since $\lambda_1>\lambda_2$, then $(\lambda_2-\lambda_1)^{-1}e_{uv}$ will be positive, which indicates that node $u$ will rotate counter-clockwise and the vector will remain in the first quadrant.

For 2(b), $e_{uv}<0$  and $\lambda_1<\lambda_2$. $(\lambda_2-\lambda_1)^{-1}e_{uv}$ will be negative, which indicates that node $u$ will rotate clockwise and the vector will be in the fourth quadrant.

\end{proof}

Note that the spectral coordinates are non-negative for all Perron Frobenius eigenvectors before perturbations, so that the above proof holds true. However, in real applications, iteration methods could some times solve non-positive Perron Frobenius eigenvectors, so the observed rotation direction is reversed. Therefore, in general, for $\lambda_1<\lambda_2$, $u$ will rotate away from those nodes of $C_2$, but $u$ will rotate towards those nodes of $C_2$ on the other hand.

\begin{figure}
\centering
\includegraphics[height=55mm,width=55mm]{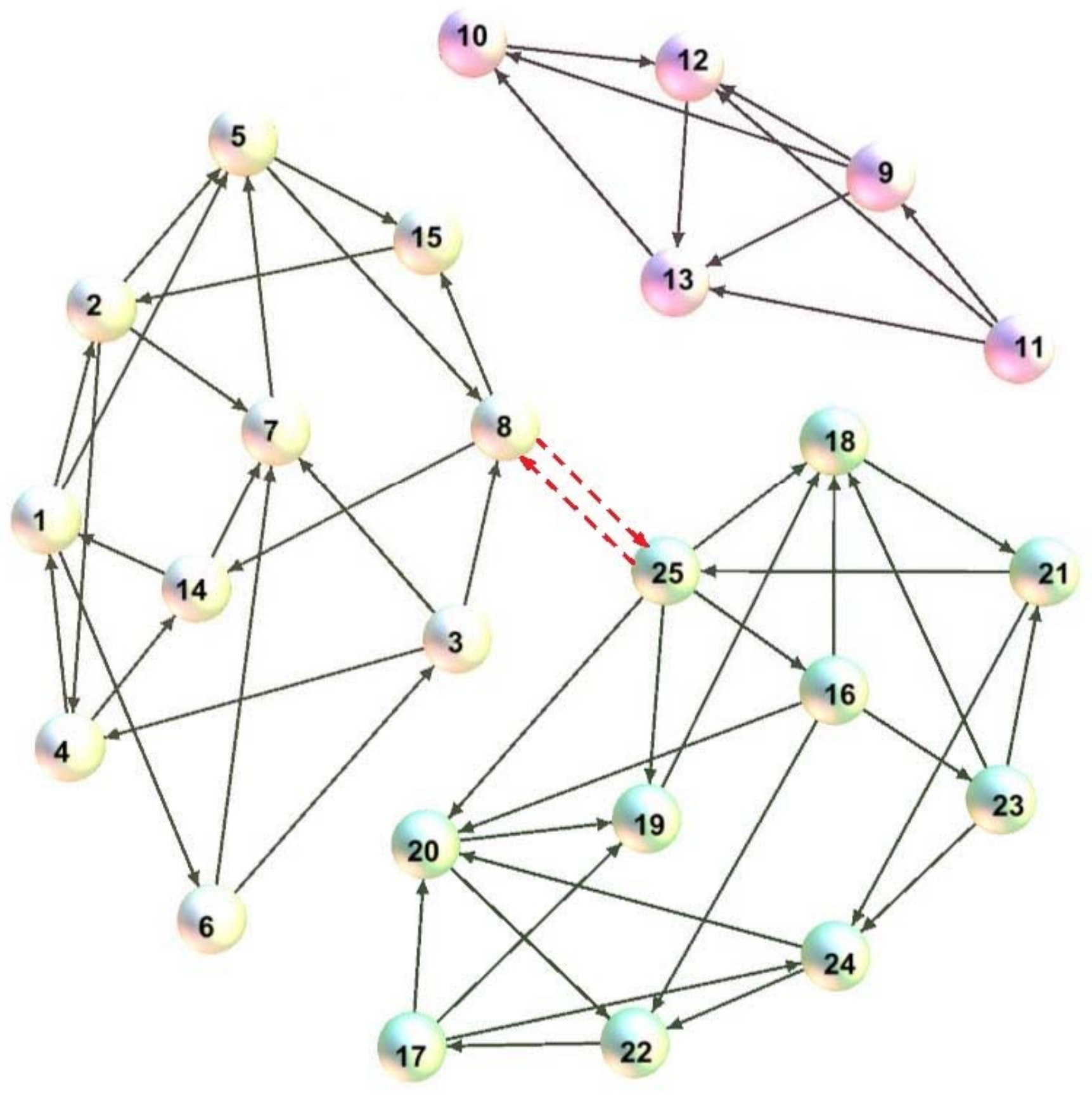}
\caption{Example graph with 3 communities, where node 8 and 25 are connected by negative or positive edges.}
\label{fig:3com}
\end{figure}

\begin{figure}[ht]
\resetsubfigs
\centering
\subfigure[Adding 8$\rightarrow$25 positive]{%
\includegraphics[height=34mm,width=38mm]{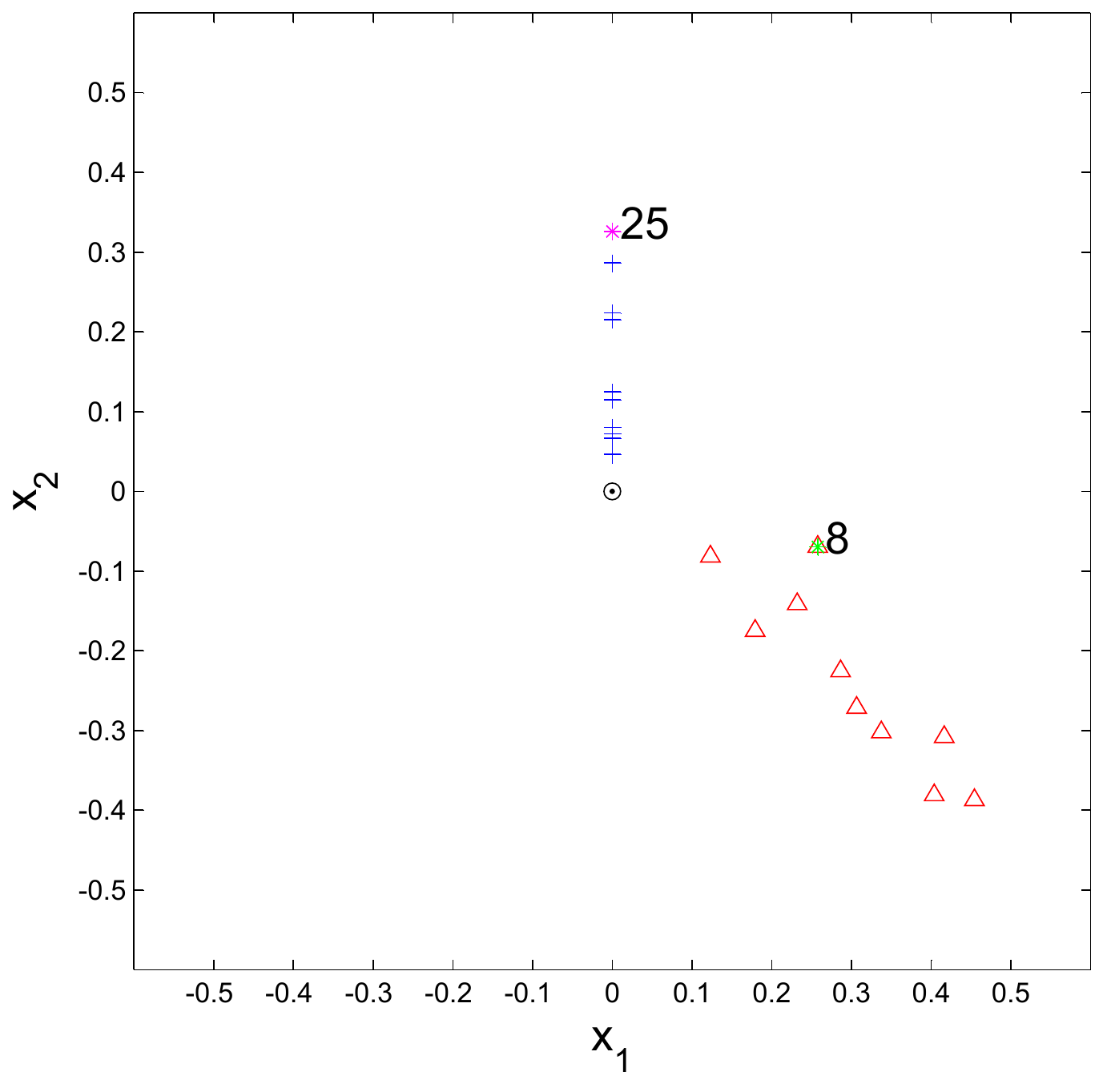}
\label{fig:825}}
\quad
\subfigure[Adding 8$\rightarrow$25 negative]{%
\includegraphics[height=34mm,width=38mm]{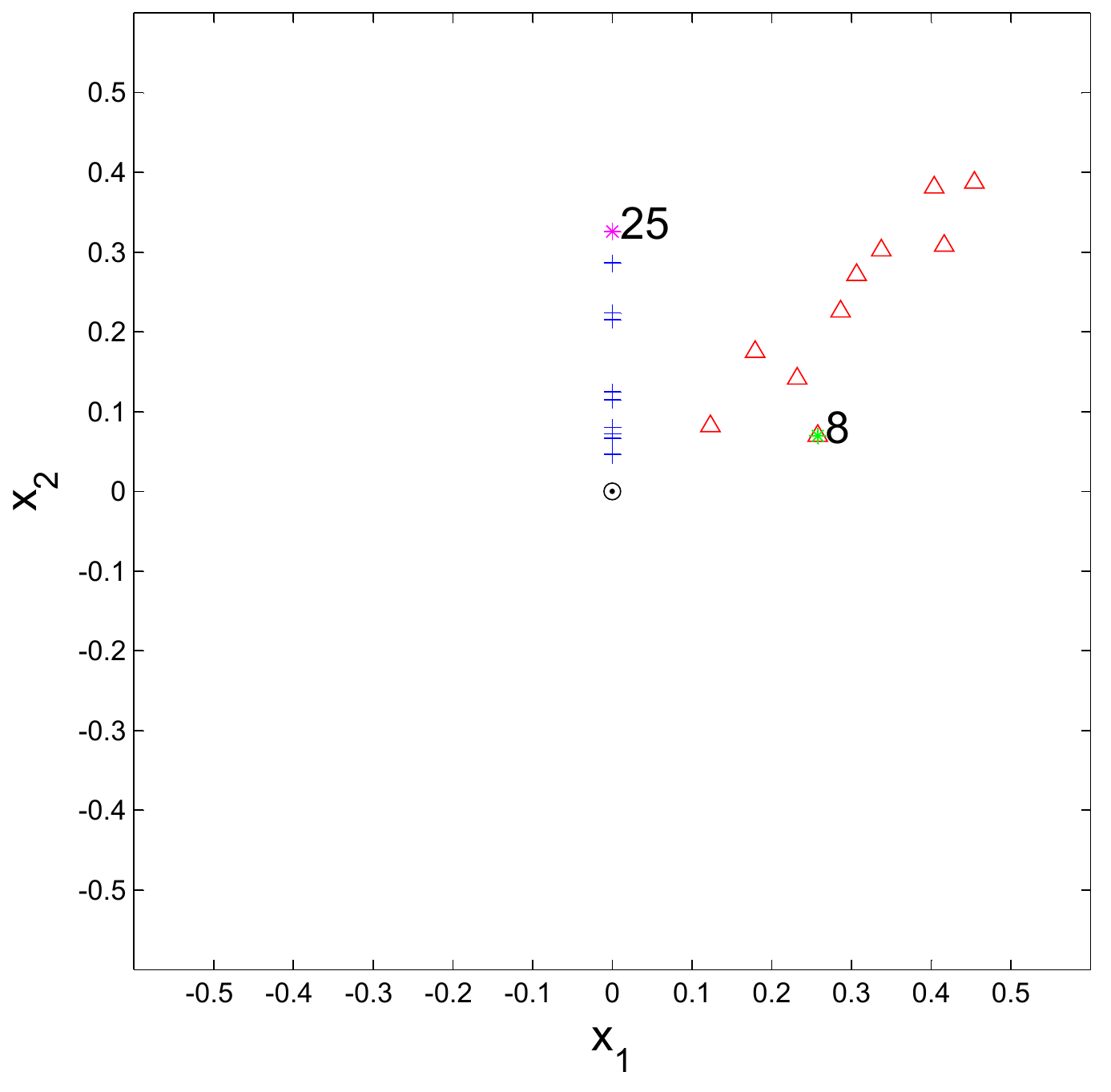}
\label{fig:825neg}}
\subfigure[Adding 25$\rightarrow$8 positive]{%
\includegraphics[height=34mm,width=38mm]{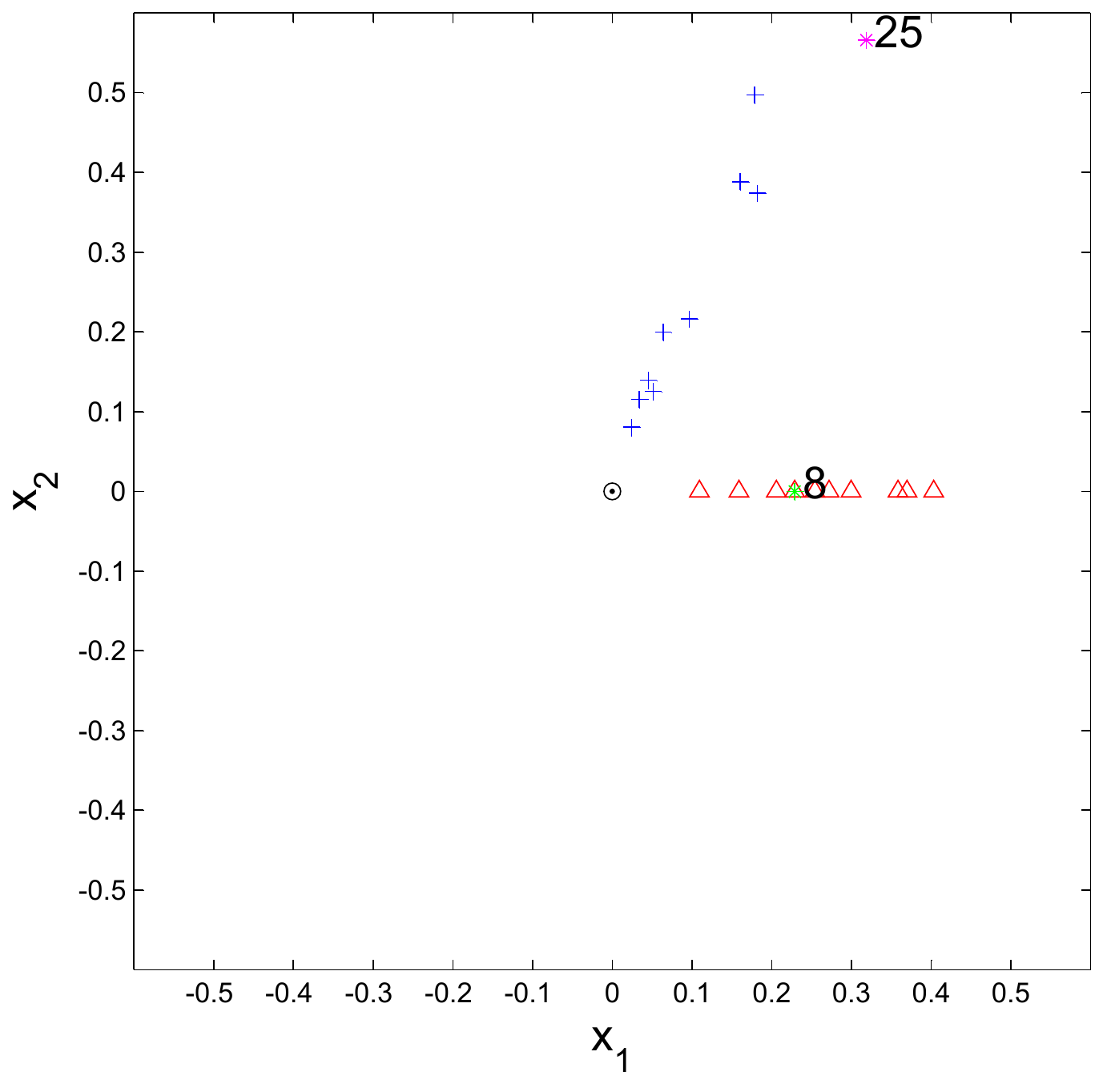}
\label{fig:258}}
\quad
\subfigure[Adding 25$\rightarrow$8 negative]{%
\includegraphics[height=34mm,width=38mm]{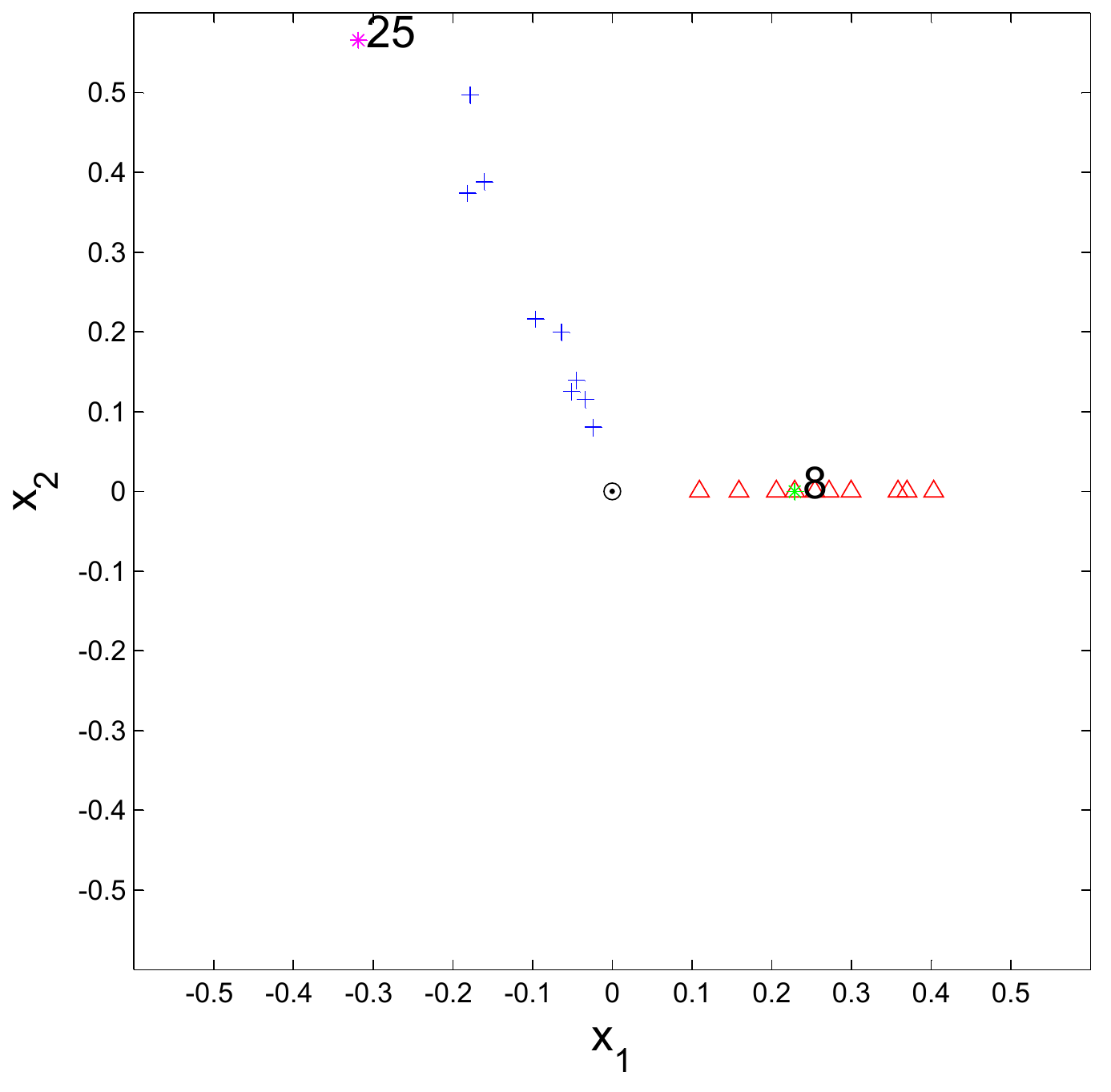}
\label{fig:258neg}}
\subfigure[Adding 8$\leftrightarrow$25 positive]{%
\includegraphics[height=34mm,width=38mm]{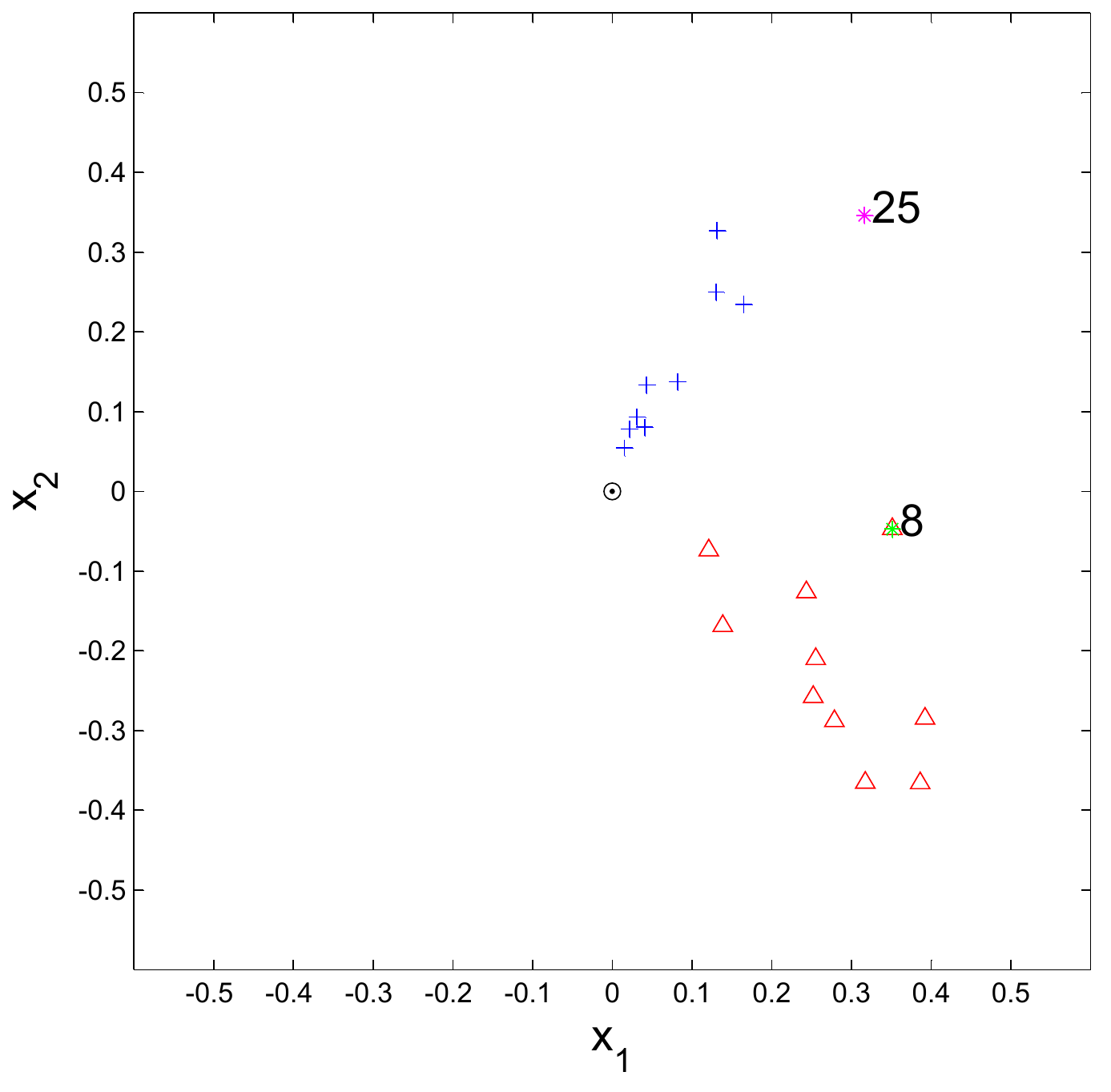}
\label{fig:825bi}}
\quad
\subfigure[Adding 8$\leftrightarrow$25 negative]{%
\includegraphics[height=34mm,width=38mm]{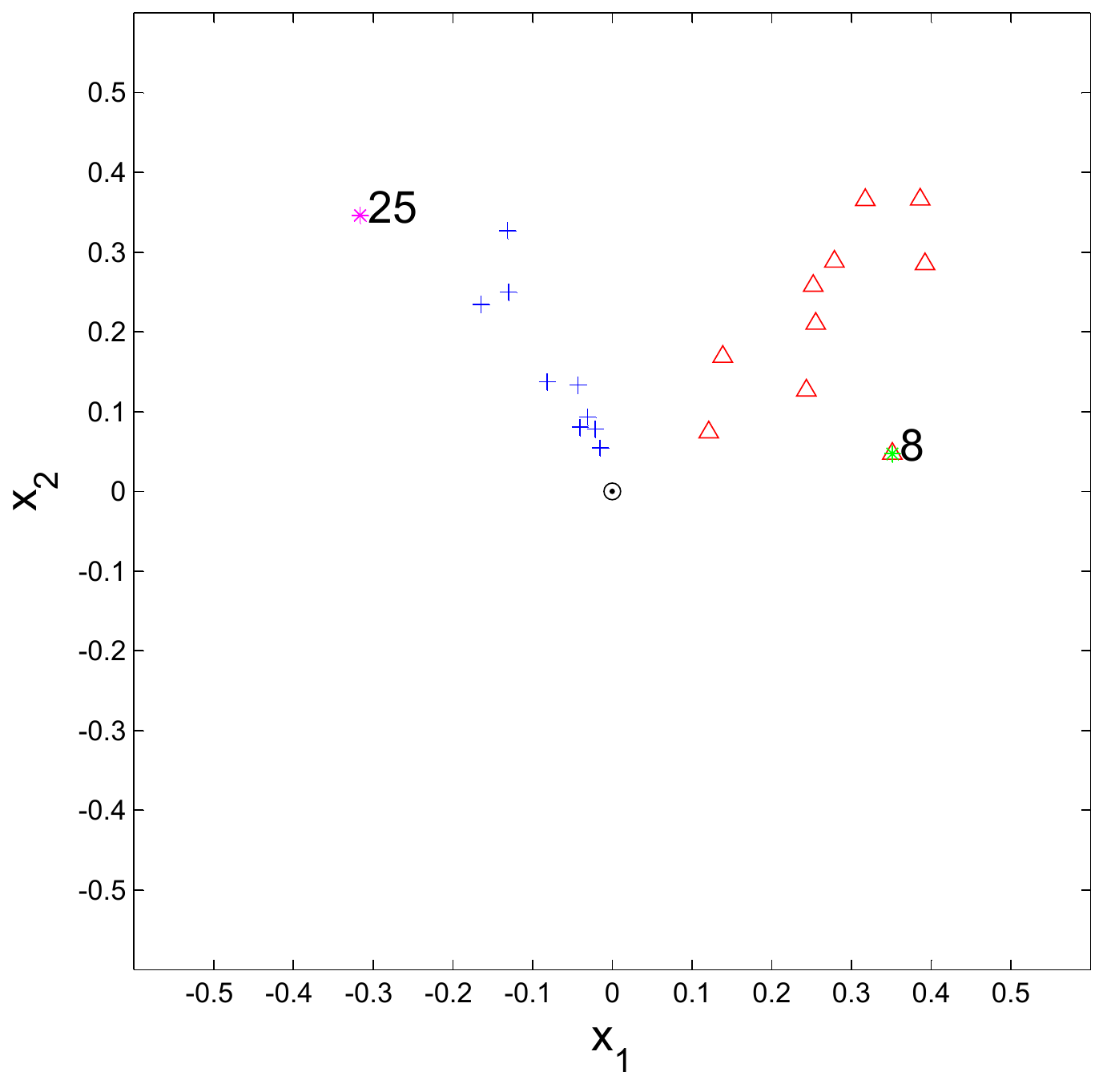}
\label{fig:825bi_neg}}
\caption{Spectral Coordinates of Nodes under Perturbation}
\label{FIG:pos_neg_inter2}
\end{figure}

{\bf \noindent Illustrative Example}. We introduce a toy graph as in Figure \ref{fig:3com} where nodes 8 and 25 are chosen to perform the illustration.
Figure \ref{FIG:pos_neg_inter2}  illustrates the rotations with respect of perturbation edge directions and signs in the spectral space. The  triangles represent nodes from cluster C1, labeled with  1-8 and 15 in Figure \ref{fig:3com}, while the crosses represent nodes from cluster C2, labeled with 16-25. Node 8 is marked with green and node 25 is marked with magenta in order to separate from other nodes. The Perron Frobenius eigenvalue is 1.8839 for cluster C1 and 1.7284 for cluster C2, so that $\lambda_1>\lambda_2$. The sub-figures on the left hand side correspond to positive perturbation, with edge $8 \rightarrow 25$, $8 \leftarrow 25$ , and $8 \leftrightarrow 25$, and those on the right side correspond to the negative perturbation respectively. All the observations match our Theorem \ref{thm:Generalized_Rotation}. For example, Figure \ref{fig:825} shows node 8 and other nodes in C1 rotate clockwisely while node 25 and other nodes in C2 stay on the original line with a positive edge $8 \rightarrow 25$, which matches our result 1(a) in Theorem \ref{thm:Generalized_Rotation}. Similarly, Figure \ref{fig:825neg} shows node 8 and other C1 nodes rotate anti-clockwisely while node 25 and other nodes in C2 stay on the original line with a negative edge $8 \rightarrow 25$, which matches our result 2(a) in Theorem \ref{thm:Generalized_Rotation}. Figures \ref{fig:258} and \ref{fig:258neg} show the effect due to edge directionality. Figures  \ref{fig:825bi} and \ref{fig:825bi_neg} show the combined effects of both directions.

We would emphasize that Figure \ref{fig:825bi} demonstrates the exact rotation phenomena as discovered and examined in \cite{Wu2011} where the spectral analysis of undirected unsigned graphs was conducted.
Our work is the first to discover the relationship of the rotations with the signed edge directions and signs and to explain such phenomena in theory in a much more general setting.

\subsection{Spectral Analysis of Intra Cluster Perturbation}\label{sect:Intra}

In general, the subgraph for each cluster is treated as an intra cluster perturbation from a nonnegative subgraph $A_i$ such that $\widetilde{A_i}=A_i+E_i$, with $E_i$ containing all negative intra cluster edges. This process can be treated as a transition from a nonnegative graph with Perron Frobenius property into a signed graph with uncertain properties. When negative intra cluster perturbations occur, depending on the amount and locations of negative edges added, the perturbed spectral space could be categorized into 3 different types: the perturbed spectral space still has Perron Frobenius property, the spectral radius is an eigenvalue but the associated eigenvector is no longer positive, and spectral radius is complex with complex eigenvectors.

\begin{Definition}\label{def:PFn}
PFn is defined to be the set of matrices having Perron Frobenius property:
The dominant eigenvalue $\lambda_1$ of each component $C_i$ is positive, simple and the corresponding eigenvector $\mathbf{x_1}$ is positive.
\end{Definition}

The Perron Frobenius eigenpair for each cluster will exist for unsigned graphs as long as a large strongly connected core exists. As a result, those real eigenpairs could be used for guiding the clustering process. However, for signed graphs, negative entries within clusters may cause those Perron roots for their characteristic polynomials to change drastically. Furthermore, the corresponding eigenvectors and spectral projections will also change accordingly. Since the coefficients of the polynomial in Equation \ref{eq:Char_Poly} are determined by the determinant $|A-\lambda*I|_{det}$, which is calculated iteratively with the entries of $A$, the polynomial itself could be either increasing, decreasing, concave or convex. Therefore, the resulting eigenvalues could be positive real, negative real, zero or complex.

\begin{Definition}\label{def:Primitive_Irriducible}
A nonnegative matrix $A$ is called primitive, if there exist some real positive number $M$ such that $A^M$ is positive. If the associated adjacency matrix $A$ for a given graph $G$ is nonnegative irreducible, then $A$ is primitive.
\end{Definition}

\begin{Definition}\label{def:Eventually_Pos}
A matrix $A$ is called eventually positive, if there exist some real positive number $M_0$ such that $A^m$ is positive for all $m\geq M_0$.
\end{Definition}

\begin{Definition}\label{def:WPFn}
WPFn is defined to be the set of matrices having weak Perron Frobenius property:
The dominant eigenvalue $\lambda_1$ is positive and the corresponding eigenvector $\mathbf{x_1}$ is nonnegative.
\end{Definition}

According to the extension to Perron Frobenius theorem, when $A$ is primitive, the Perron Frobenius eigenvalue $r$ is real and strictly bigger than any other eigenvalue $\lambda$ in absolute value. Furthermore, the associated eigenvector is positive. This extension leads to a series of studies related to matrices belonging to the PFn set.

In \cite{zaslavsky1999jordan}, the authors showed that the set of eventually positive square real matrices is equivalent to the set of those having Perron Frobenius properties. Therefore, a very important implication  is that: if a cluster as a subgraph has the eventually positive properties, then it belongs to the PFn set.
However, we do not have any explicit results to show how small the negative entries should be in order for a graph to retain Perron Frobenius property.
Therefore, we can only have the following two results for DSGs: First, depending on the density and location of the negative edges, the Perron Frobenius eigenpair may disappear as the perturbed graph may not belong to WPFn. Second, in the case where the spectral radius is still an eigenvalue, the associated eigenvector may no longer be positive.

In summary, when the perturbed spectral space still has a real eigenvalue as its spectral radius, all the theoretical results from the previous sections and work \cite{li2015analysis} will hold. Therefore, the spectral space properties of observed signed graphs behave exactly the same as those of nonnegative graphs. With an increasing number of negative edges added to a graph that belong to PFn, there should exist a certain threshold above which we have a complex eigenvalue as the adjacency matrix's spectral radius.

\begin{table*}[!htb]
\centering
\begin{tiny}
\caption{Statistics of synthetic data and partition quality}
\begin{tabular}{|c||c|c|c|c|c|c||c|c|c|c|c|c|} \hline
        \multirow{2}{*}{Dataset}    &\multicolumn{2}{|c|}{Edge/$+$ratio/$-$ratio}& \multirow{2}{*}{$k$}   &\multirow{2}{*}{DBI}  & \multirow{2}{*}{$Q$}  & \multirow{2}{*}{Angle} & \multicolumn{6}{|c|}{\textsc{Accuracy(\%)}} \\\cline{2-3}\cline{8-13}
           &Intra &Inter& & & & & SC-SDG &SC-DSG-M & SC-DSG-Re&AugAdj &UniAdj &  SNCut         \\\hline
    \emph{Syn-1}    &67653/0.4/0   & 144283/0.2/0   & 5 &0.1745 &0.2770   & $89.3^\circ$  &100&100 & 100 &100  &100 &100\\\hline
    \emph{Syn-2}   &67588/0.4/0   & 144335/0.1/0.1   & 5 &0.4711 &1.9141   & $92.2^\circ$  &100&100 & 100 &100 &100 &100\\\hline
    \emph{Syn-3}   &67545/0.4/0   & 144362/0/0.2   & 5 &0.0926 &-1.0798   & $90.1^\circ$  &100&100 & 100  &100 &100 &100\\\hline
    \emph{Syn-4}   &67618/0.4/0   & 400420/0.7/0   & 4 &1.9774 &0.0321   & $76.5^\circ$  &\textbf{72.9}   &68.3 &  68.8 &70.3 &71.9    &62\\\hline
    \emph{Syn-5}   &80749/0.4/0.08   & 144294/0.2/0   & 5 &0.4290 &0.2221   & $87.9^\circ$  &100 &100 & 100 &100 &100 &100\\\hline
    \emph{Syn-6}  &81019/0.4/0.08   & 144372/0/0.2   & 5 &0.3827  &0.4442   & $89.1^\circ$  &100 &100 &100 &100 &100 &100\\\hline
    \emph{Syn-7}   &80789/0.4/0.08   & 438193/0.4/0  & 4 &1.4309 & 0.0776   & $82.0^\circ$  &\textbf{92}    &\textbf{92}&  \textbf{92} &\textbf{92} &\textbf{92}   &89.9\\\hline\hline
    \emph{Syn-8}   &101002/0.4/0.16   & 144220/0.2/0   & 5 &1.9958 &0.0749   & $76.8^\circ$ & \textbf{67.5}    &\textbf{62.9}&  \textbf{62.9} &\textbf{62.9} &\textbf{62.9}    &65.1\\\hline
    \emph{Syn-9}    &127448/0.4/0.36   & 144283/0.1/0.1   & 5 &2.9701 & -0.2260   & $82.4^\circ$  &\textbf{59.3}  &57.1 & 58.9 &n/a &55.1    &56\\\hline

\end{tabular}
\label{table:SynthDirectedSigned}
\end{tiny}
\end{table*}

\section{Spectral Clustering for DSGs}\label{sect:Tests}

The results from Section \ref{sect:Inter} described how node spectral coordinates are changed due to inter cluster  perturbation, while the results from Section \ref{sect:Intra} described the potential complex eigenpairs due to  negative intra cluster perturbation. With the two results combined, we have a full picture of the spectral properties of DSGs.

\subsection{Algorithm}

Based on the perturbed Perron Frobenius simple invariant subspace results, we present our spectral clustering based graph partition algorithm  (SC-DSG)  for directed signed graphs.
A key challenge is to determine the proper eigenvectors for clustering from a set of potentially complex eigenvectors.

\begin{algorithm}[htb]
\caption{\textit{SC-DSG}: Spectral Clustering for Directed Signed Graphs} \label{alg:General-Adj-Cluster}%
\textbf{Input:} $A, \tau, \alpha$~ \\
\textbf{Output:} cluster number $K$, clustering result $CL$~\\
\begin{algorithmic}[1]
 \STATE Compute eigenvectors of $A$ corresponding to the largest $\tau$ eigenvalues $\Lambda$ in magnitude, and denote the set $D$;
 \STATE Normalize the eigenvectors $\bar{\coord}_u = \frac{\coord_u}{\|\coord_u\|}$;
 \STATE $C \leftarrow$ real eigenvectors from the set $D$ with same signed components;
 \STATE $K \leftarrow Cardinality(C)$;
 \STATE $M\leftarrow -\inf$;
 \FOR {each $c \in \emptyset \cup D\setminus C$}
    \IF {$c$ is complex}
        \STATE $c \leftarrow$ split into [Re($c$) Im($c$)]
    \ENDIF
    \STATE Apply \kmeans on $C \cup c$ to get clustering result $CL_{temp}$ of $K$ clusters;
    \STATE Compute the signed modularity scores $M_{temp}$;
    \IF {$M_{temp} \geq \alpha M$ ($\alpha \in [0,1]$ adjusts the objective function)}
        \STATE $K \leftarrow K+1$;
        \STATE $C \leftarrow C \cup c$;
        \STATE $CL \leftarrow CL_{temp}$;
        \STATE $M \leftarrow M_{temp}$;
    \ENDIF
 \ENDFOR
 \STATE $K \leftarrow K-1$;
 \STATE Return number of clusters $K$ and clustering result $CL$;
\end{algorithmic}
\end{algorithm}

Algorithm \ref{alg:General-Adj-Cluster} includes the following major steps: computing eigenvectors of the adjacency matrix; normalization of the eigenvectors; selecting the initial set of eigenvectors with same signed components; splitting complex eigenvectors into real and imaginary parts;
projection of the nodes onto a unit sphere; clustering the nodes according to their location on the unit sphere using the classic k-means clustering algorithm; screen all the potential eigenpairs based on the signed modularity to find meaningful partitions.

{\bf \noindent Handling Complex Eigenpairs}.
Our algorithm uses  eigenvectors corresponding to the spectral radii of each component. However, due to negative intra cluster perturbations, the eigenpairs corresponding to the spectral radii of the clusters may not be real anymore. As a result, the Perron Frobenius invariant subspace tends to diminish in DSGs. The problem is how to define the Euclidean topology in the complex coordinate space formed by the complex basis.
The k-means clustering used in most of the spectral based clustering methods could not produce meaningful results in the coordinate space of $\mathbb{C}^n$, since the Euclidean distance of two complex coordinates with only imaginary part will be negative. In our algorithm, we propose to split the complex spectral coordinate space corresponding to the Perron Frobenius invariant subspace into a higher dimensional real coordinate space and apply the k-means method in the real space.

Since the vector space of $\mathbb{C}^n$ is an isomorphism of $\mathbb{R}^{2n}$, the complex portion of the simple invariant subspace could be split into two parts that contain the real and imaginary parts of the original complex coordinates respectively.
Assume that the perturbed simple invariant subspace is of dimension $K$ and there are $C$ complex eigenvectors, then the newly formed spectral space will have $K+C$ dimensions. However, we still want to form $K$ clusters in this space, since we are embedding only $K$ communities into this space.

Our algorithm uses the stepwise forward strategy to find a set of eigenvectors to maximize the signed modularity \cite{anchuri2012communities}.  For a graph with two clusters, it has the form:
\begin{multline}
\label{eq:Signed_Lap}
Q_s=\sum_{i,j\in C_1}(P_{ij}-\frac{d_{p_i}d_{p_j}}{2m_p})+\sum_{i,j\in C_2}(P_{ij}-\frac{d_{p_i}d_{p_j}}{2m_p}) \\
+\sum_{i\in C_1,j\in C_2}(N_{ij}-\frac{d_{n_i}d_{n_j}}{2m_n})+\sum_{i\in C_2,j\in C_1}(N_{ij}-\frac{d_{n_i}d_{n_j}}{2m_n}),
\end{multline} where $\frac{d_id_j}{2m}$ is the expected number of edges between nodes i and j, $P_{ij}=\frac{A+|A|}{2}$ and $N_{ij}=\frac{A-|A|}{2}$.

A good partition seeks to achieve the signed modularity score as high as possible.  However, exhaustive search is infeasible to identify the optimal set because it is NP-hard.
However, as pointed out in the work \cite{lancichinetti2011limits}, modularity has two limitations in community detection: the tendency to merge small subgraphs and the tendency to split large subgraphs. Both of those drawbacks could cause the network partition to produce incorrect results. We  include a parameter $\alpha$ so that the importance of the objective function could be linearly adjusted, since some partitions could still be meaningful even if the resulting signed modularity score is slightly lower. As shown in the work \cite{li2015analysis}, it can be tuned down to detect substructures of clusters.

\begin{table*}[!htb]
\centering
\begin{small}
\caption{Real Data Statistics and Results \label{table:RealStat}}
\begin{tabular}{|c||c|c|c||c|c|c|c|} \hline
        \multirow{2}{*}{Algorithm} &\multicolumn{3}{|c||}{\textsc{Network Statistics}}  & \multicolumn{4}{|c|}{\textsc{Signed Modularity/Clusters}} \\\cline{2-4}\cline{5-8}
        & Nodes&+Edges&-Edges&SC-DSG &  AugAdj &UniAdj &  SNCut        \\\hline
    \emph{Sampson's} & 18& 97&87 & \textbf{2.52}/3  & 1.4200/2 &  -0.8757/15  & -1.0503/11 \\\hline
    \emph{Slashdot}  &79120 & 370234& 117517& \textbf{0.16}/6  & 0.1512/4 & 0.155/5 & 0.1072/22  \\\hline
    \emph{Wikisigned}&138592 & 650653&89744 &  \textbf{0.1734}/5 & 0.0785/3 & 0.0848/5 & 0.0789/37   \\\hline
    \emph{Epinion}   &131828 &717667 &123705 & \textbf{0.3416}/5  & 0.337/6 & -0.174/5 & 0.2595/13  \\\hline
\end{tabular}
\end{small}
\end{table*}

\section{Experiment}

\subsection{Baseline Algorithms}

In our experiment, we compare our SC-DSG with the following state-of-the-art baseline algorithms:  1) The Augmented$\_$ADJ (AugAdj) \cite{li2015analysis} is an adjacency based spectral clustering method for unsigned directed graph; 2) UniAdj \cite{Wu2014} is an adjacency based spectral clustering method for signed undirected graphs; 3) The signed normalized cut (SNcut) \cite{kunegis2010spectral} is an improved version of the signed Laplacian method where weighting schemes are adjusted to form better partitions; 4) SC-DSG-M is a variant of SC-DSG and only uses  the modulii of the eigenvector entries as spectral coordinates; and 5) SC-DSG-Re is another variant of SC-DSG and only uses the real part of the eigenvector entries as spectral coordinates.

AugAdj  can be used to deal with the directed signed graph as it ignores the use of any complex eigenpairs. Both UniAdj and SNcut require symmetric adjacency matrices as input. In our experiment, we build the symmetrized versions of the original directed graphs by the following process: $A_{ij}=A_{ij}=-1$ if either$A_{ij}=-1$ or $A_{ji}=-1$, $A_{ij}=A_{ij}=1$ if either $A_{ij}=1$ or $A_{ji}=1$, and $A_{ij}=A_{ji}=0$ otherwise.
We limit the search for each method to 50 eigenpairs. Signed modularity, DBI and average angles between clusters in the spectral projection space are reported in addition to accuracy. Note that isolated clusters are orthogonal.
The two variants of SC-DSG are used to demonstrate the usefulness of incorporating the whole complex eigenvectors in the clustering.

\subsection{Synthetic Data}

We generate 9 synthetic graphes. Each graph has 5 clusters with 240,220,200,180 and 160 nodes.
The edges for Syn-1 to Syn-9 are generated using uniform random distribution.
The column ``Intra" of table \ref{table:SynthDirectedSigned} shows the number of intra cluster edges as well as the positive and negative densities.
In particular, the $x/y/z$  means that there are $x$ intra cluster edges and the density of positive (negative) edges is $y\times 100$ ($z\times 100$) percent.
Similarly, the column ``Inter" shows the corresponding statistics for inter cluster edges. The negative edges for both intra and inter clusters are injected into the 5-block graph so that the perturbed graphs possess desired structural properties.

For synthetic data, Syn-1 to Syn-4, there are 40\% positive edges but no negative edges within clusters. The inter cluster positive edge density of Syn-1 is 0.2 and there is no negative inter cluster edge. This is the classic case of the Perron Frobenius clusters. All the methods achieved 100\%  accuracy. The average cluster angle is close to 90 degrees.
 In Syn-2, half of the inter cluster edges from Syn-1 are converted into negative edges. The clusters  belong to the PFn set, so the perturbed Perron Frobenius invariant subspace contains the real eigenvectors corresponding to the spectral radii of the clusters. The average angle between clusters is 92.2 degrees.
In Syn-3, all the inter cluster edges are negative.  All methods achieve 100 percent accuracy for this case and the average angle is 90.1 degrees.
In Syn-4, the inter cluster positive edge density is increased to 0.7 without the inter cluster negative edge. In this setting, all methods report 4 clusters, where the accuracies drop to around 60 to 70 percent. Since we have dense inter cluster connections, the results are expected.

For Syn-5, Syn-6 and Syn-7, the positive (negative) intra cluster edge density is 0.4 (0.08).  For Syn-5, the positive inter cluster density is 0.2 with no negative inter cluster edge.  All methods achieve 100 percent accuracy. For Syn-6,  the negative inter cluster edge density is 0.2. Since the inter cluster contains only negative edge,  all methods still identify the clusters 100 percent. For Syn-7, more  positive inter cluster edges are added. The partition accuracy drops. Same as Syn-4, only 4 clusters are detected. Our algorithm, SC-DSG, achieve the best accuracy among all methods.

For Syn-8, the negative intra cluster perturbation is doubled to 0.16. The positive intra cluster edge density remains as 0.4. The positive inter cluster edge density is set to be 0.2. Since the PFn properties no longer exist, the accuracy values drop by over 20 percent for all  methods. There exist some complex eigenvalues  whose modulii equal the spectral radius.
For Syn-9, the negative intra cluster edge density is further increased to 0.36. Both the positive and negative inter cluster density is 0.1. This is the most complex case for directed signed graphs. With no surprise, all clustering methods perform poorly, with SC-DSG achieving the best accuracy (59.3\%).

 To summarize,  when under small inter cluster perturbations (as the cases for Syn-1 to Syn-3), as long as clusters satisfy the Perron Frobenius property, all methods  perform the same, since the correct perturbed Perron Frobenius simple invariant subspace is captured by all methods.  When under small inter cluster perturbations (as the cases of Syn-5 and Syn-6), as long as clusters satisfy the Perron Frobenius property, all  methods can still achieve the correct partition. As demonstrated in Syn-4 and Syn-7, dense inter cluster edges cause clusters to merge, so the clustering accuracies decrease. As demonstrated in Syn-8 and Syn-9,  when the Perron Frobenius property begins to disappear, the clustering accuracy will decrease  more. In all cases, SC-DSG achieves the best accuracy results for the synthetic experiments.

\subsection{Real Data}

In this section we conduct our empirical experiments on four real directed signed graphs, Sampson's, Wikisigned, Slashdot Zoo and Epinion.
Sampson's work \cite{konect:sampson} contains the opinions of 18 trainee monks about their relationships towards each other during the period of time when the clique fell apart. Each monk was asked to rate others from 1 to 3 based on like or dislike. Later on, the responses were converted into an binary signed adjacency matrix.
Slashdot \cite{kunegis2009slashdot} is a technology news site where users can mark others as ``friend" or ``foe" and influence scores seen by them. Therefore, the entire network could be seen as a trust network.
Wikisigned \cite{konect:2016:wikisigned-k2} contains interactions between the users of the English Wikipedia that have edited pages about politics. Each interaction, such as text editing, reverts, restores and votes are given a positive or negative value. Epinion \cite{leskovec2010signed} is an online product rating website. The users can choose to trust or distrust others and self vote is allowed. As a result, the network could be viewed as a trustworthy relationship network.

Table \ref{table:RealStat} shows the graph statistics, the signed modularity and number of clusters reported by each method.
We can see that our method achieves the best signed modularity value for all four datasets. We observe that the eigenvector associated with cluster 6 is complex for Slashdot, the eigenvector associated with cluster 4 is complex, all the others are real. Note that we cannot report accuracy because of no ground truth about these four datasets.

\section{Related Work}

Spectral methods have been well investigated and successfully adopted in solving graph or network structure related problems
such as community partition \cite{Newman2006a,VonLuxburg2007}, anomaly detection \cite{chandola2009anomaly}, link prediction \cite{al2011survey}. In this section, we focus on related work on spectral analysis of signed networks from applied mathematics and linear algebra.

The theoretical results of graph theories pertaining the spectral properties of general signed graphs are relatively scare. From the first proposition of the Perron Frobenius theorem for nonnegative irreducible square matrices in 1912 \cite{frobenius1912matrizen} to the most recent works \cite{friedland1978inverse,noutsos2006perron,zaslavsky1999jordan} that extended the result into eventually irreducible nonnegative directed graphs a few decades ago, many problems related to signed graphs remain open. However, the results deduced from eventually irreducible nonnegative signed graphs provided us some important insights in handling general signed graphs, which aid us in the process of analyzing the spectral properties of clusters with intra negative edges.

The authors in the work \cite{kunegis2014handbook} pointed out that three properties can be read off the complex eigenvalues: whether a graph is nearly acyclic, whether a graph is nearly symmetric, and whether a graph is nearly bipartite. If a directed graph is acyclic, its adjacency matrix is nilpotent and therefore all its eigenvalues are zero as shown in the book \cite{cvetkovic1995spectra} (page 81).
The complex eigenvalue plot can therefore serve as a test for networks that are nearly acyclic. When a directed network is symmetric, the adjacency matrix $A$ is symmetric and all its eigenvalues are real. As a result, a directed network close to symmetric has complex eigenvalues near the real line. Additionally, the eigenvalues of an undirected bipartite signed graph are symmetric around the imaginary axis, so the amount of symmetric along the imaginary axis can serve as an indicator for bipartivity.

\section{Conclusion}
In this work, we conducted spectral analysis of directed signed graphs. By using matrix perturbation theory, we derived the approximations of the spectral coordinates of nodes in the spectral projection space formed by perturbed Perron Frobenius invariant subspace and explained the effects of added intra and inter cluster edges to the spectral coordinates. A spectral clustering algorithm for directed signed graphs, SC-DSG,  was proposed according to the theoretical results and was tested on both synthetic and real datasets. The results demonstrated the effectiveness of the algorithm.

Although we took the first step in analyzing the adjacency eigenspace of directed signed graphs by generalizing its properties in relation to the underlying network structures, there are still many open problems in this area. Especially, the lacking of the appropriate linear algebra and graph theory tools to fully describe the cluster properties outside of the PFn set makes further extending the current results a very difficult task. However, many future research topics could be built upon our current theoretical framework, such as fraud detection, dynamic network analysis, and signed network embedding. We will also study the scalability of our algorithm. We would emphasize that our algorithm has the same bottleneck, the eigen decomposition, as all other spectral clustering methods. However, out method should have slightly better performance than Laplacian based methods because the adjacency matrix is often sparser than the Lapalician matrix.

\section{Acknowledgments}

This work was supported in part by U.S. National Science Foundation under awards (1564039 and 1564250).

%\bibliographystyle{IEEEtran}
%{\footnotesize
%\bibliography{mergedBIB}}

\end{document}

%% file: macros.tex
\usepackage{times}
\usepackage{graphicx}
\usepackage{epsf}
\usepackage{verbatim}
\usepackage{cite}
\usepackage{url}
\usepackage{color}
\usepackage{alltt}
\usepackage{psfrag}
\usepackage{latexsym}

\newcommand{\Comment}[1]{}

\newcommand{\SmallSpace}{\vspace*{-1.5ex}}